\documentclass[12pt]{article}
\usepackage{amsmath}
\usepackage{graphicx,psfrag,epsf}
\usepackage{enumerate}
\usepackage{natbib}
\usepackage{url} 
\usepackage{verbatim}
\usepackage[]{hyperref}
\usepackage{sectsty, secdot}
\usepackage{amssymb}
\usepackage[utf8]{inputenc}
\usepackage{array,epsfig,fancyheadings,rotating}

\newcommand{\blind}{0}

\newtheorem{theorem}{Theorem}

\newenvironment{proof}[1][Proof]{\noindent\textbf{#1.} }{\ \rule{0.5em}{0.5em}}

\begin{document}

\def\spacingset#1{\renewcommand{\baselinestretch}%
{#1}\small\normalsize} \spacingset{1}


\if0\blind
{
  \title{\bf Bayes' Theorem under Conditional Independence}
  \author{Jun Hu\footnote{Jun Hu is Assistant Professor in the Department of Mathematics and Statistics, Oakland University, Rochester, MI 48309. Email address: \url{junhu@oakland.edu}.} \ and \ Xianggui Qu\footnote{Xianggui Qu is Professor in the Department of Mathematics and Statistics, Oakland University, Rochester, MI 48309. Email address: \url{qu@oakland.edu}.}
      }
    \date{}
  \maketitle
} \fi

\if1\blind
{
  \bigskip
  \bigskip
  \bigskip
  \begin{center}
    {\LARGE\bf Title}
\end{center}
  \medskip
} \fi

\bigskip
\begin{abstract}
In this article we provide a substantial discussion on the statistical concept of conditional independence, which is not routinely mentioned in most elementary statistics and mathematical statistics textbooks. Under the assumption of conditional independence, an extended version of Bayes' Theorem is then proposed with illustrations from both hypothetical and real-world examples of disease diagnosis.      
\end{abstract}

\noindent%
\vfill

\newpage
\spacingset{1.45} 
\section{Introduction}
\label{Sect. 1}
Inarguably, \textit{conditional probability} and \textit{independence} are two concepts that play an important role in statistical theory. Most elementary statistics and mathematical statistics textbooks discuss these two concepts in detail and then illustrate the well-known Bayes' Theorem, such as \cite{Wackerly Mendenhall and Scheaffer (2014)} and \cite{Hogg et al. (2015)}. To our surprise, however, the concept of \textit{conditional independence} has been rarely mentioned since its appearance in \cite{Dawid (1979)} more than forty years ago, let alone a systematic introduction.

In this article, therefore, we give a substantial discussion on conditional independence. We focus on conditional independence of events instead of random variables for illustrative purposes. This way, the concept is made as simple as possible for students to understand, but no simpler. In Section \ref{Sect. 2}, a number of straightforward examples are provided to point out some basic properties of conditional independence as well as a series of seemingly correct yet wrong arguments that students may make to supplement the existing literature. Then in Section \ref{Sect. 3}, we propose an extended version of Bayes' Theorem under the assumption of conditional independence to accommodate practical applicability, and also use hypothetical and real-world examples to demonstrate the possible application in disease diagnosis. The materials will be helpful for motivating undergraduate students to explore the story-line with more depth-confidence-grasp about how to apply the impressing result efficiently. We end with some concluding thoughts in  Section \ref{Sect. 4}.

\section{Conditional Independence}\label{Sect. 2}

In this section, we first revisit (statistical) independence between two events and thereby introduce the concept of conditional independence. After a sequence of preliminary results are presented, we extend the idea from the two-event case to multiple-event case.

\subsection{Basic concepts and preliminary results}

\noindent \textit{Definition 1.} (Independence) Two events $A_1$ and $A_2$ are said to be \textbf{independent} if and only if
\begin{eqnarray}\label{1}
P(A_1 | A_2) = P(A_1),
\end{eqnarray}
provided that $P(A_2)>0$.

Alternatively, independence can also be defined as follows:

\noindent \textit{Definition 2.} (Independence) Two events $A_1$ and $A_2$ are said to be \textbf{independent} if and only if 
\begin{eqnarray}\label{2}
P(A_1 \cap A_2) = P(A_1) \cdot P(A_2).
\end{eqnarray}

The first definition is straightforward to convey the meaning of independence: if two events are independent, then knowledge that one of the events has occurred has no effect on the probability that the other will occur. Nevertheless, most students will find the second definition more favorable since it does not require the assumption that $P(A_2)>0$ to make the conditional probability well defined. That is why we are going to introduce conditional independence along the line with the equation \eqref{2}.

\noindent \textit{Definition 3.} (Conditional Independence) Two events $A_1$ and $A_2$ are said to be \textbf{conditionally independent} given event $B$ with $P(B)>0$, if and only if
\begin{eqnarray}\label{3}
P(A_1 \cap A_2 | B) = P(A_1|B) \cdot P(A_2|B).
\end{eqnarray}
Otherwise, we say events $A_1$ and $A_2$ are \textbf{conditionally dependent} given $B$.

Conditional independence of two events can be interpreted in view of Definition 1: Under the condition that event $B$ has occurred, event $A_1$ (or $A_2$) occurring does not affect the probability that event $A_2$ (or $A_1$) occurs. Naturally, students may ask how independence and conditional independence might be associated with each other. Here, we provide several crucial remarks with examples to answer this question, which also demonstrate that independence and conditional independence can behave quite differently. We believe this will help students to avoid making misleading arguments that seem to make sense at the first glance. Afterwards, students may like to further scrutinize those arguments and construct their own counter-examples for practice. Throughout the article, $S$ is used to denote the sample space with \textit{equally likely} outcomes without otherwise specified, and the complement of an event $A$ is represented by $A^\prime$. 

\noindent \textit{Remark 1.} \textbf{Independence does not imply conditional independence necessarily, and vice versa.}

\noindent \textit{Example 1.} Let $S = \{1,2,3,4,5,6\}$. Define three events $A_1 = \{1,2,3\}$, $A_2=\{2,4\}$ and $B=\{1,3,4\}$. By the assumption of equally likely outcomes in $S$, it is trivial for students to obtain that
$$P(A_1)=\frac{1}{2},P(A_2)=\frac{1}{3},P(A_1\cap A_2)=\frac{1}{6},P(A_1|B)=\frac{2}{3},P(A_2|B)=\frac{1}{3},P(A_1\cap A_2|B)=0.$$
Therefore, students immediately find that $A_1$ and $A_2$ are independent since $P(A_1\cap A_2)=P(A_1)P(A_2)$. However, $A_1$ and $A_2$ are not conditionally independent given $B$ due to the fact that $P(A_1|B)P(A_2|B)\ne P(A_1\cap A_2|B)$.

\noindent \textit{Example 2.} Let $S = \{1,2,3,4,5,6,7,8\}$. Define three events $A_1 = \{1,2,3\}$, $A_2=\{2,4\}$ and $B=\{1,2,3,4,5,6\}$. Clearly,
$$P(A_1)=\frac{3}{8},P(A_2)=\frac{1}{4},P(A_1\cap A_2)=\frac{1}{8},P(A_1|B)=\frac{1}{2},P(A_2|B)=\frac{1}{3},P(A_1\cap A_2|B)=\frac{1}{6}.$$
Hence, $A_1$ and $A_2$ are conditionally independent under $B$, but are not independent of each other.

\noindent \textit{Remark 2.} \textbf{That two events $A_1$ and $A_2$ are conditionally independent given event $B$ does not necessarily imply that $A_1$ and $A_2$ are also conditionally independent given $B^\prime$, the complement of $B$.}

\noindent \textit{Example 3.} Let $S = \{1,2,3,4,5,6,7,8\}$. Define three events $A_1 = \{1,3,5,7\}$, $A_2=\{2,5,8\}$ and $B=\{1,2,3,4,5,6\}$. Then, by noting that
$$P(A_1|B)=\frac{1}{2},P(A_2|B)=\frac{1}{3}, \ \mbox{and} \ P(A_1\cap A_2|B)=\frac{1}{6},$$
one has that $A_1$ and $A_2$ are conditionally independent under $B$. However, $A_1$ and $A_2$ are not conditionally independent under $B^\prime$ since
$$P(A_1|C^\prime)=\frac{1}{2}, P(A_2|C^\prime)=\frac{1}{2}, \ \mbox{but} \ P(A_1\cap A_2|B^\prime)=0.$$

\noindent \textit{Remark 3.} \textbf{That two events $A_1$ and $A_2$ are both independent and conditionally independent given event $B$ does not imply that $A_1$ and $A_2$ are conditionally independent given $B^\prime$.}

\noindent \textit{Example 4.} Let $S = \{1,2,...,16\}$. Define three events as follows: $A_1 = \{1,2,...,11,12\}$, $A_2=\{1,2,3,4,5,6,15,16\}$ and $B=\{6,7,8,13,14,15\}$, for which 
$$A_1\cap A_2=\{1,2,3,4,5,6\} \ \mbox{and} \ B^\prime=\{1,2,3,4,5,9,10,11,12,16\}.$$
It is not hard for students to work out the following quantities:
\begin{eqnarray*}
&&P(A_1)=\frac{3}{4},P(A_2)=\frac{1}{2},P(A_1\cap A_2)=\frac{3}{8}=P(A_1)P(A_2),\\
&&P(A_1|B)=\frac{1}{2},P(A_2|B)=\frac{1}{3}, P(A_1\cap A_2|B)=\frac{1}{6}=P(A_1|B)P(A_2|B),\\
&&P(A_1|B^\prime)=\frac{9}{10},P(A_2|B^\prime)=\frac{3}{5},P(A_1\cap A_2|B^\prime)=\frac{1}{2} \ne P(A_1|B^\prime)P(A_2|B^\prime).
\end{eqnarray*}
In this example, students will notice that $A_1$ and $A_2$ are independent and conditionally independent under $B$, but they are conditionally dependent under $B^\prime$.

\noindent \textit{Remark 4.} \textbf{That two events $A_1$ and $A_2$ are both conditionally independent given $B$ and conditionally independent given $B^\prime$ does not necessarily imply $A_1$ and $A_2$ are independent.}

\noindent \textit{Example 5.} Let $S = \{1,2,...,14\}$. Define three events $A_1 = \{1,2,3,9,10,11,12,13,14\}$, $A_2=\{1,6,7,12,13,14\}$ and $B=\{1,2,3,4,5,6\}$. Then, we have
\begin{eqnarray*}
&&P(A_1|B)=\frac{1}{2},P(A_2|B)=\frac{1}{3}, P(A_1\cap A_2|B)=\frac{1}{6}=P(A_1|B)P(A_2|B),\\
&&P(A_1|B^\prime)=\frac{3}{4},P(A_2|B^\prime)=\frac{1}{2},P(A_1\cap A_2|B^\prime)=\frac{3}{8} = P(A_1|B^\prime)P(A_2|B^\prime),\\
&&P(A_1)=\frac{9}{14}, P(A_2)=\frac{3}{7}, P(A_1\cap A_2)=\frac{2}{7} \ne P(A_1)P(A_2).
\end{eqnarray*}
Thus, while $A_1$ and $A_2$ are conditionally independent under either $B$ or $B^\prime$, $A_1$ and $A_2$ are not independent.

Next, the following theorem points out a possible association between independence and conditional independence.
\begin{theorem}\label{Thm. 1}
Let $A_1$, $A_2$ and $B$ be three events with $P(B)>0$. If $A_1$ is independent of $B$ and $A_1$ is also independent of $A_2\cap B$, then $A_1$ and $A_2$ are conditionally independent given $B$.
\end{theorem}

\begin{proof}
By checking the definition of conditional independence between two events, students can establish the identity that
\begin{eqnarray}\label{4}
P(A_1\cap A_2|B) = \frac{P(A_1\cap A_2 \cap B)}{P(B)} = P(A_1) \cdot \frac{P(A_2\cap B)}{P(B)} = P(A_1|B) \cdot P(A_2|B).
\end{eqnarray}
Hence, the statement holds.
\end{proof}



\begin{theorem}\label{Thm. 2}
Given event $B$ with $P(B)>0$, the following four statements in terms of events $A_1,A_2$ and their complements are equivalent: (i) $A_1$ and $A_2$ are conditionally independent; (ii) $A_1^\prime$ and $A_2$ are conditionally independent; (iii) $A_1$ and $A_2^\prime$ are conditionally independent; (iv) $A_1^\prime$ and $A_2^\prime$ are conditionally independent.
\end{theorem}

\begin{proof}
We show that (i)$\Rightarrow$(ii)$\Rightarrow$(iv)$\Rightarrow$(iii)$\Rightarrow$(i). First, we show (i)$\Rightarrow$(ii). 
Note that when $A_1$ and $A_2$ are conditionally independent given $B$, one has
\begin{equation}\label{5}
\begin{split}
P(A_1^\prime \cap A_2 | B) & = \frac{P(A_1^\prime \cap A_2 \cap B)}{P(B)} = \frac{P(A_2 \cap B) - P(A_1 \cap A_2 \cap B)}{P(B)} \\
& = P(A_2|B) - P(A_1|B)P(A_2|B) = P(A_1^\prime | B)P(A_2|B),
\end{split}
\end{equation}
which indicates that $A_1^\prime$ and $A_2$ are conditionally independent given $B$. Note that students will need to recall the definition of conditional probability and the identity that $P(A_1|B)+P(A^\prime|B)=1$ to claim \eqref{5}. Following this result, (iv) holds immediately by retaining the first event $A_1^\prime$ and substituting the second event $A_2$ with $A_2^\prime$, as how we moved forward from (i) to (ii). In the same manner, (iv)$\Rightarrow$(iii) and (iii)$\Rightarrow$(i) can also be justified together with the interchangeability of $A_1$ and $A_2$.   
\end{proof}



\subsection{From two events to multiple events}

In analogy to pairwise and mutual independence of multiple events, we are now in a position to generalize the notion of conditional independence of multiple events.

\noindent \textit{Definition 4.} (Pairwise and Mutual Conditional Independence) A collection of events $A_1,A_2,...,A_n(n \ge 3)$ is said to be \textbf{pairwise conditionally independent} given event $B$ with $P(B)>0$, if and only if for all $i \ne j$,
\begin{eqnarray}\label{PCI}
P(A_i \cap A_j|B) = P(A_i|B)\cdot P(A_j|B). 
\end{eqnarray}
A collection of events $A_1,A_2,...,A_n(n \ge 3)$ is said to be \textbf{mutually conditionally independent} given another event $B$ with $P(B)>0$, if and only if for every subset of indices $i_1,i_2,...,i_k$,
\begin{eqnarray}\label{MCI}
P(A_{i_1} \cap A_{i_2} \cap \cdots \cap A_{i_k}|B) = P(A_{i_1}|B)\cdot P(A_{i_2}|B) \cdots P(A_{i_k}|B). 
\end{eqnarray}

For convenience, we drop the modifier ``mutually'' when talking about multiple mutually conditionally independent events in practice. Hence, whenever we say that $A_1,...,A_n$ are ``conditionally independent'', we mean ``mutually conditionally independent.'' Students may take it as an exercise to give examples showing that Remarks 1-4 are also satisfied for multiple conditionally independent events. In this case, Theorem \ref{Thm. 2} can also be modified accordingly.

\begin{theorem}\label{Thm. 3}
Given a collection of events $A_1,A_2,...,A_n (n \ge 2)$, let $A_i^\ast$ be either $A_i$ or its complement $A_i^\prime$, $i=1,2,...,n$. Then, all the following statements are equivalent: $A_1^\ast,A_2^\ast,...,A_n^\ast$ are conditionally independent given event $B$ with $P(B)>0$.
\end{theorem}

\begin{proof}
One may start with the assumption that $A_1,...,A_n$ are conditionally independent under $B$, and show that the collection of events stay conditionally independent if we substitute one of them with its complement, for instance, $A_1^\prime,A_2,...,A_n$. This can be done in a similar way as we proved (i)$\Rightarrow$(ii) in Theorem \ref{Thm. 2}. Then, we use this result repeatedly with one $A_i^\ast$ replaced by its complement at a time, and a complete proof will go through. We leave out many details for brevity.
\end{proof}

\section{Extending Bayes' Theorem}\label{Sect. 3}

When it comes to conditional probability, Bayes' Theorem is helpful for reversing the role of the event and the condition. Suppose $A$ is an event with $P(A)>0$, and $B_1,B_2,...,B_m (m \ge 2)$ are mutually exclusive and exhaustive events, that is, a partition of the sample space $S$. Then,
\begin{eqnarray}\label{BT}
P(B_k|A) = \frac{P(A|B_k)P(B_k)}{\sum_{i=1}^{m}P(A|B_i)P(B_i)}, \ \ k=1,2,...,m.
\end{eqnarray}

Considering the set of events $\{B,B^\prime\}$ as a trivial partition of $S$, we have a simplified version of Bayes' Theorem as follows:
\begin{eqnarray}\label{SBT}
P(B|A) = \frac{P(A|B)P(B)}{P(A|B)P(B)+P(A|B^\prime)P(B^\prime)},
\end{eqnarray} 
which is widely used in diagnostic testing for diseases. See the example below.

\noindent \textit{Example 6.} Let $D$ be the event that a (rare) disease is present, so $D^\prime$ denotes the event that the disease is not present. Suppose there exists a diagnostic test for this disease, and let $T^+$ and $T^-$ be the events that the test result is positive and negative, respectively. Here,
\begin{enumerate}
  \item $P(D)$, called the \textit{prevalence}, is interpreted as the probability that a randomly-selected person has the disease and is assumed known.
  \item $P(T^+|D)$, called the test \textit{sensitivity}, is interpreted as the probability that the test gives a ``true positive'' result. As a characteristic of the test, it is known to us.
  \item $P(T^-|D^\prime)$, called the test \textit{specificity}, is interpreted as the probability that the test gives a ``true negative'' result. As another characteristic of the test, it is also known to us. 
  \item $P(D|T^+)$, called the \textit{positive predictive value} (PPV), is the conditional probability that one has the disease given that the test result is positive. If the test is positive and the PPV is high enough, then it would be appropriate to initiate a treatment. On the other hand, if the PPV is low, then further testing might be appropriate.
  \item $P(D^\prime|T^-)$, called the \textit{negative predictive value} (NPV), is the conditional probability that one does not have the disease given that the test result is negative. If the test is negative and the NPV is high enough, then one can conclude no disease is present. On the other hand, if is low, then further testing might be appropriate.
\end{enumerate}
One may refer to \citet{Altman and Bland (1994a),Altman and Bland (1994b)} for more details of these notions. 

Mostly, we are interested in the PPV. Based on Bayes' Theorem in \eqref{SBT}, we substitute $A$ with $T^+$, $B$ with $D$ and obtain
\begin{equation}\label{6}
\begin{split}
P(D|T^+) & = \frac{P(T^+|D)P(D)}{P(T^+|D)P(D)+P(T^+|D^\prime)P(D^\prime)}\\
& = \frac{P(T^+|D)P(D)}{P(T^+|D)P(D)+[1-P(T^-|D^\prime)][1-P(D)]}.
\end{split}
\end{equation}
Again, students need to recall the fact that $T^+$ and $T^-$ are complementary events and thus $P(T^+|D^\prime)+P(T^-|D^\prime)=1$. 

Most textbook examples stop discussions upon the derivation of PPV, even when it is sufficiently small indicating the necessity of further testing. However, students may be curious about the following questions: What if a second test is conducted and the test result is still positive, or negative? At that point, what is the probability that one has the disease, indeed?    

In this section, we are ready to extend Bayes' Theorem under the assumption of conditional independence and answer the above questions.

\subsection{An extended Bayes' Theorem}\label{3.1}

Provided a set of events $\{B_1,B_2,...,B_m,m\ge2\}$ with all positive probabilities, which forms a partition of the sample space $S$, suppose events $A_1,A_2,...,A_n,n\ge2$ are conditionally independent under each $B_k,k=1,2,...,m$. Suppose also that we are interested in the conditional probability $P(B_k|\bigcap_{i=1}^{n}A_i)$. For any $k=1,2,...,m$ and $i=1,2,...,n$, if the quantities
$P(B_k)$'s and $P(A_i|B_k)$'s are all known to us, we give the so-called \textit{extended Bayes' Theorem} as follows:
\begin{theorem}[Extended Bayes' Theorem]\label{EBT}
\begin{eqnarray}\label{7}
P\left(B_k|\bigcap_{i=1}^{n}A_i\right) = \frac{P(B_k)\cdot\prod_{i=1}^{n}P(A_i|B_k)}{\sum_{k=1}^{m} P(B_k)\cdot\prod_{i=1}^{n}P(A_i|B_k)}.
\end{eqnarray}
\end{theorem}
\begin{proof}
By the definition of conditional probability, students can easily obtain
\begin{eqnarray}\label{8}
P\left(B_k|\bigcap_{i=1}^{n}A_i\right)=\frac{P\left(\bigcap_{i=1}^{n} A_i \cap B_k\right)}{P\left(\bigcap_{i=1}^{n} A_i\right)},
\end{eqnarray}
where the numerator
\begin{eqnarray}\label{9}
P\left(\bigcap_{i=1}^{n} A_i \cap B_k\right) = P\left(\bigcap_{i=1}^{n} A_i | B_k\right)P(B_k) = P(B_k)\cdot\prod_{i=1}^{n}P(A_i|B_k),
\end{eqnarray}
and the denominator
\begin{equation}\label{10}
\begin{split}
P\left(\bigcap_{i=1}^{n} A_i\right) & = \sum_{k=1}^{m}P\left(\bigcap_{i=1}^{n} A_i | B_k \right) \cdot P(B_k)\\
& = \sum_{k=1}^{m} P(B_k)\cdot\prod_{i=1}^{n}P(A_i|B_k).
\end{split}
\end{equation}
The proof is now complete by combining \eqref{9} and \eqref{10} together.
\end{proof}

\noindent \textit{Remark 5.} In terms of $P\left(\bigcap_{i=1}^{n} A_i\right)$ in \eqref{10}, a possible error that some students may make is to treat $A_i$'s as independent events and thus write
$$P\left(\bigcap_{i=1}^{n} A_i\right)=\prod_{i=1}^{n}P(A_i),$$
where $P(A_i),i=1,2,...,n$ is further computed by using the Law of Total Probability:
$$P(A_i) = \sum_{k=1}^{m}P(A_i|B_k)P(B_k).$$ 
As is pointed out in Remark 4, however, this is not necessarily true. And Example 5 provides a simple counter-example when $m=n=2$. It emphasizes that one should not confuse independence with conditional independence.

The significance of Theorem \ref{EBT} is immediately recognized in answering questions raised in Example 6. Suppose a person whose first test for the disease is positive, denoted by $T_1^+$, goes for a second test separately and the test is still positive, denoted by $T_2^+$. Due to the test sensitivity and specificity, it is reasonable to assume that $T_1^+$ and $T_2^+$ are conditionally independent under $D$ as well as under $D^\prime$.  Then, according to the extended Bayes' Theorem in Theorem \ref{EBT}, the probability that he actually has the disease can be updated as follows:
\begin{eqnarray}\label{11}
P(D|T_1^+ \cap T_2^+) = \frac{P(T^+|D)^2P(D)}{P(T^+|D)^2P(D)+[1-P(T^-|D^\prime)]^2(1-P(D))},
\end{eqnarray}
where $P(D),P(T^+|D)$ and $P(T^-|D^\prime)$ continue to denote prevalence, sensitivity and specificity mentioned earlier, respectively. If $P(D|T_1^+\cap T_2^+)$ is still low, then a third test might be appropriate. In general, we can obtain the probability that one has the disease given $n$ conditionally independent positive test results:
\begin{eqnarray}\label{PPV}
P\left(D|\bigcap_{i=1}^{n} T_i^+\right) = \frac{P(T^+|D)^n P(D)}{P(T^+|D)^nP(D)+[1-P(T^-|D^\prime)]^n(1-P(D))}.
\end{eqnarray}
This can be left as an exercise for students to practice.  

\noindent \textit{Remark 6.} For an accurate diagnostic test, both sensitivity and specificity are close to one. Then, it is safe to assume that the quantity $$\frac{P(T^+|D)}{1-P(T^-|D^\prime)},$$
defined as the \textit{likelihood ratio} \citep[See][]{Altman and Bland (1994b)}, is larger than 1. As a result, it is not hard for students to observe that
$$\lim_{n\to\infty}P\left(D|\bigcap_{i=1}^{n} T_i^+\right)=1$$
by using some elementary calculus techniques, which implies that a sequence of positive tests can be a good indicator of the presence of disease. 

\subsection{A hypothetical example}

To illustrate the application of the extended Bayes' Theorem and Remark 6, we include a hypothetical example borrowed from \citet[p. 220]{Utts and Heckard (2011)} that is appealing to students taking elementary statistics courses with modifications.

\noindent \textit{Example 7.} Last week, Alicia went to her physician for a routine medical exam and was told that one of her tests came back positive, indicating that she may have a disease $D$. It is known that the test is 95\% accurate as to whether someone has this disease or not. In other words, the test sensitivity and specificity are both 95\%. Suppose that only 1 out of 1000 women of Alicia's age indeed has $D$. With knowledge on Bayes' Theorem, Alicia then computed her actual chance of having the disease $D$ given the positive test result by referring to \eqref{6}:
\begin{eqnarray}\label{12}
P(D|T^+) = \frac{(0.95)(0.001)}{(0.95)(0.001)+(1-0.95)(1-0.001)} = 0.019.
\end{eqnarray} 
The positive predicted value is so small that further testing for the disease $D$ may be needed. Therefore, Alicia went for the same test for $D$ for a second time. Unfortunately, the test result turned out positive again. At this point, by using the extended Bayes' Theorem in \eqref{11}, we have
\begin{eqnarray}\label{12}
P(D|T_1^+ \cap T_2^+) = \frac{(0.95)^2(0.001)}{(0.95)^2(0.001)+(1-0.95)^2(1-0.001)} = 0.265.
\end{eqnarray}
With a second positive test result, Alicia's chance of having the disease increased hugely by almost 14 times. Suppose Alicia took a third and fourth test and they were again positive. Referring to \eqref{PPV}, we have
\begin{eqnarray}
P(D|T_1^+ \cap T_2^+ \cap T_3^+) = \frac{(0.95)^3(0.001)}{(0.95)^3(0.001)+(1-0.95)^3(1-0.001)} = 0.873,
\end{eqnarray}
and   
\begin{eqnarray}
P(D|T_1^+ \cap T_2^+ \cap T_3^+\cap T_4^+) = \frac{(0.95)^4(0.001)}{(0.95)^4(0.001)+(1-0.95)^4(1-0.001)} = 0.992,
\end{eqnarray}
closer and closer to 1.

\subsection{A real data illustration}
\cite{Bhatti and Wightman (2008)} provided a real-world application of Bayes' Theorem. Table 2 in their paper gives the probabilities of being HIV positive for one and two positive tests with sensitivity $0.99$ and specificity $0.99$ with various prevalence in ten geographic regions. In the spirit of their paper, we calculate the probabilities of adult aged 15 to 49 being HIV positive for one, two, and three positive tests using our extended Bayes' Theorem based on the data coming from the \cite{Joint United Nations Programme on HIV/AIDS (2018)}. The results are presented in Table 1.

\begin{table}[h!]\label{Tab. 1}
\small
\begin{center}
\caption{Probability of adult aged 15 to 49 being HIV positive by geographic region given one positive test, two and three conditionally independently positive tests with sensitivity $0.99$ and specificity $0.99.$}
\begin{tabular}{lcccc}\\\hline\hline
 & Adult & One & Two  & Three \\
Region       & Prevalence & Positive & Positives & Positives\\\hline
Asia and the Pacific & 0.002 & 0.1656 & 0.9516 & 0.9995\\
Caribbean &0.012&0.5460&0.9917&0.9999\\
Eastern and Southern Africa & 0.070 & 0.8817 & 0.9986 & 1.0000\\
Eastern Europe and Central Asia &0.009&0.4734&0.9889&0.9999\\
Latin America &0.004& 0.2845 & 0.9752& 0.9997\\
Middle East and North Africa &0.001& 0.0902 & 0.9075 & 0.9990\\
Western and Central Africa & 0.015 & 0.6012 & 0.9933 & 0.9999\\
Western and Central Europe and North America & 0.002 & 0.1656 & 0.9516 & 0.9995\\
\hline
\end{tabular}
\end{center}
\end{table}
For small prevalence (e.g., 0.001), the PPV given one positive test may remain to be small (e.g., 0.0902) even if both sensitivity and specificity are large (e.g., 0.99). Given a second positive test, however, this conditional probability will increase dramatically and approach 1. All probabilities of adult aged 15 to 49 being HIV positive for three positive tests are almost equal to $1.$ The real data illustration has justified Remark 6. 

Furthermore, it will be a good idea for instructors to interpret the interesting phenomenon of small PPV in detail: This is due to the low prevalence of disease instead of the ``inaccurate'' diagnostic test. It demonstrates the necessity of follow-up confirmatory tests. And in fact, the probability $P\left(D|\bigcap_{i=1}^{n} T_i^+\right)$ approaches 1 very fast when both sensitivity and specificity are large enough, showing the great significance of diagnostic test accuracy.          

\subsection{Applications}\label{Sect 3.4}

In this section, we propose a sequential testing scheme in which the extended Bayes' Theorem is applied for more efficient disease diagnosis. For $n \ge 1$, define $p_n = P(D|\bigcap_{i=1}^{n}T_i^\ast)$, where $T_i^\ast = T_i^+ \text{ or } T_i^-$ meaning that the $i$th test is positive or negative, $i=1,...,n$, so $p_n$ can be interpreted as the conditional probability that one has the disease given a sequence of test results $\{T^\ast_n\}$. Let $\{\alpha_n\}$ and $\{\beta_n\}$ be two nondecreasing series of numbers predetermined appropriately such that 
$$0<\alpha_1\le \cdots \le \alpha_n \le \cdots \le \beta_1 \le \cdots \le \beta_n \le \cdots <1.$$ Then, we develop a stopping rule for diagnostic testing as follows:
\begin{eqnarray}\label{dtsr}
N = \inf\{n\ge1: p_n \le \alpha_n \text{ or } p_n \ge \beta_n\}.
\end{eqnarray}
That is, we conduct the test successively and terminate at the first time $N=n$ such that either $p_n\le \alpha_n$ or $p_n \ge \beta_n$ happens. And we conclude that the disease is present (or not present) if $p_N \ge \beta_N$ (or $p_N \le \alpha_N$). Students from some interdisciplinary programs may find it interesting to follow this direction and explore the possibility for future research work. 

\section{Overall Concluding Thoughts}\label{Sect. 4}

In Section \ref{Sect. 2}, we have discussed conditional independence of events alone. It is worth mentioning that we can also generalize the concept of conditional independence of random variables, which is of great importance in the area of Bayesian statistics. A lot of details are left out in this article for brevity, as it is prepared for study of elementary statistics and mathematical statistics at the undergraduate level overall. One may see a batch of articles including \cite{Dawid (1979)}, \cite{Dawid (1998)} and \cite{Basu and Pereira (2011)} for reference.

Under the assumption of conditional independence, we have put forward the extended Bayes' Theorem and address its application in diagnostic testing with examples and real data illustrations. A novel idea is proposed in Section \ref{Sect 3.4} briefly, but one may follow this direction to make it more substantial. Indeed, instructors are encouraged to introduce these materials accordingly to those students standing out in class. 
\bigskip

\bibhang=1.7pc
\bibsep=2pt
\renewcommand\bibname{\large \bf References}
\expandafter\ifx\csname
natexlab\endcsname\relax\def\natexlab#1{#1}\fi
\expandafter\ifx\csname url\endcsname\relax
  \def\url#1{\texttt{#1}}\fi
\expandafter\ifx\csname urlprefix\endcsname\relax\def\urlprefix{URL}\fi

\end{document}